\PassOptionsToPackage{unicode}{hyperref}
\PassOptionsToPackage{hyphens}{url}
\PassOptionsToPackage{dvipsnames,svgnames,x11names}{xcolor}
\documentclass[
  12pt]{article}

\usepackage{amsmath,amssymb}
\usepackage{amsthm} 
\usepackage{dsfont} 
\newtheorem{Theorem}{Theorem}
\newtheorem{Proposition}{Proposition}%
\newtheorem{Lemma}{Lemma}
\newtheorem{Assumption}{Assumption}
\newtheorem{Corollary}{Corollary}

\usepackage{multirow}
\usepackage{xr} 
\usepackage{xcite} 
\usepackage{iftex}
\ifPDFTeX
  \usepackage[T1]{fontenc}
  \usepackage[utf8]{inputenc}
  \usepackage{textcomp} 
\else 
  \usepackage{unicode-math}
  \defaultfontfeatures{Scale=MatchLowercase}
  \defaultfontfeatures[\rmfamily]{Ligatures=TeX,Scale=1}
\fi
\usepackage{lmodern}
\ifPDFTeX\else  
\fi
\IfFileExists{upquote.sty}{\usepackage{upquote}}{}
\IfFileExists{microtype.sty}{
  \usepackage[]{microtype}
  \UseMicrotypeSet[protrusion]{basicmath} 
}{}
\makeatletter
\@ifundefined{KOMAClassName}{
  \IfFileExists{parskip.sty}{%
    \usepackage{parskip}
  }{
    \setlength{\parindent}{0pt}
    \setlength{\parskip}{6pt plus 2pt minus 1pt}}
}{
  \KOMAoptions{parskip=half}}
\makeatother
\usepackage{xcolor}
\makeatletter
\ifx\paragraph\undefined\else
  \let\oldparagraph\paragraph
  \renewcommand{\paragraph}{
    \@ifstar
      \xxxParagraphStar
      \xxxParagraphNoStar
  }
  \newcommand{\xxxParagraphStar}[1]{\oldparagraph*{#1}\mbox{}}
  \newcommand{\xxxParagraphNoStar}[1]{\oldparagraph{#1}\mbox{}}
\fi
\ifx\subparagraph\undefined\else
  \let\oldsubparagraph\subparagraph
  \renewcommand{\subparagraph}{
    \@ifstar
      \xxxSubParagraphStar
      \xxxSubParagraphNoStar
  }
  \newcommand{\xxxSubParagraphStar}[1]{\oldsubparagraph*{#1}\mbox{}}
  \newcommand{\xxxSubParagraphNoStar}[1]{\oldsubparagraph{#1}\mbox{}}
\fi
\makeatother

\usepackage{longtable,booktabs,array}
\usepackage{calc} 
\usepackage{etoolbox}
\makeatletter
\patchcmd\longtable{\par}{\if@noskipsec\mbox{}\fi\par}{}{}
\makeatother
\IfFileExists{footnotehyper.sty}{\usepackage{footnotehyper}}{\usepackage{footnote}}
\makesavenoteenv{longtable}
\usepackage{graphicx}
\makeatletter
\def\maxwidth{\ifdim\Gin@nat@width>\linewidth\linewidth\else\Gin@nat@width\fi}
\def\maxheight{\ifdim\Gin@nat@height>\textheight\textheight\else\Gin@nat@height\fi}
\makeatother
\setkeys{Gin}{width=\maxwidth,height=\maxheight,keepaspectratio}
\makeatletter
\def\fps@figure{htbp}
\makeatother

\makeatletter
\@ifpackageloaded{caption}{}{\usepackage{caption}}
\AtBeginDocument{%
\ifdefined\contentsname
  \renewcommand*\contentsname{Table of contents}
\else
  \newcommand\contentsname{Table of contents}
\fi
\ifdefined\listfigurename
  \renewcommand*\listfigurename{List of Figures}
\else
  \newcommand\listfigurename{List of Figures}
\fi
\ifdefined\listtablename
  \renewcommand*\listtablename{List of Tables}
\else
  \newcommand\listtablename{List of Tables}
\fi
\ifdefined\figurename
  \renewcommand*\figurename{Figure}
\else
  \newcommand\figurename{Figure}
\fi
\ifdefined\tablename
  \renewcommand*\tablename{Table}
\else
  \newcommand\tablename{Table}
\fi
}
\@ifpackageloaded{float}{}{\usepackage{float}}
\floatstyle{ruled}
\@ifundefined{c@chapter}{\newfloat{codelisting}{h}{lop}}{\newfloat{codelisting}{h}{lop}[chapter]}
\floatname{codelisting}{Listing}

\makeatother
\makeatletter
\@ifpackageloaded{caption}{}{\usepackage{caption}}
\@ifpackageloaded{subcaption}{}{\usepackage{subcaption}}
\makeatother

\ifLuaTeX
  \usepackage{selnolig}  
\fi
\usepackage[]{natbib}
\usepackage{bookmark}

\IfFileExists{xurl.sty}{\usepackage{xurl}}{} 
\hypersetup{
  pdftitle={Title},
  pdfauthor={Author 1; Author 2},
  pdfkeywords={3 to 6 keywords, that do not appear in the title},
  colorlinks=true,
  linkcolor={blue},
  filecolor={Maroon},
  citecolor={Blue},
  urlcolor={Blue},
  pdfcreator={LaTeX via pandoc}}

\newcommand{\anon}{1}

\begin{document}

\def\spacingset#1{\renewcommand{\baselinestretch}%
{#1}\small\normalsize} \spacingset{1}


\if1\anon
{
  \title{\bf Graph Neural Networks for Generalized Mundlak Estimator under Network Confounding}
  \author{Lianyan Fu\thanks{
    This study was supported by the National Social Science Fund of China (24BTJ064). The authors would like to express their gratitude to these funding organizations for their financial support, which enabled the completion of this study.}\hspace{.2cm}\\
    School of Mathematics and Statistics, Liaoning University\\
    Rui Wang \\
    School of Mathematics and Statistics, Liaoning University\\
    Zihan Zhang\\
    School of Mathematics and Statistics, Liaoning University}
  \maketitle
} \fi

\if0\anon
{
  \bigskip
  \bigskip
  \bigskip
  \begin{center}
    {\LARGE\bf Graph Neural Networks for Generalized Mundlak Estimator under Network Confounding}
\end{center}
  \medskip
} \fi

\bigskip
\begin{abstract}
This paper proposes a generalized Mundlak estimator based on graph neural networks (\textit{GME-GNN}). The estimator is designed to mitigate bias arising from group-level heterogeneity and to accommodate within-group dependence among individuals. 
Traditional fixed-effects models handle group heterogeneity via group-specific intercepts, but require overly strict linear additivity and intra-group independence assumptions, and are confined to within-group comparisons. Rather than relying on intercepts, \textit{GME-GNN} uses aggregated group-level balancing statistics to fully control between-group confounding, enabling valid cross-group comparisons and relaxing linearity constraints. It further employs graph neural network message-passing to adaptively learn nonlinear representations and capture intra-group interaction effects. 
Theoretical analysis shows that the estimator satisfies double robustness and is asymptotically normal. Simulation and empirical studies confirm its performance.
\end{abstract}

\noindent%
\vfill

\newpage
\spacingset{1.75} 

\section{Introduction}
\label{sec:intro}
When estimating causal effects in observational studies, research units are typically grouped by factors such as geographical region and industry category. \citet{altonji2005selection} have documented that unobserved characteristics across these groups can substantially influence study outcomes. The classical fixed-effects model accounted for unobserved group-level heterogeneity by incorporating group-specific intercepts. However, as pointed out by \citet{arkhangelsky2024fixed}, the estimation of causal effects using fixed-effects models relied on overly strict assumptions. These assumptions included homogeneous treatment effects, linear-additive relationships between treatment variables and covariates, and complete independence among individuals within groups.

To relax these assumptions, the Mundlak estimator reinterpreted the fixed-effects model by replacing group fixed effects with group-level means of the treatment variables and covariates \citep{mundlak1978pooling}. \citet{arkhangelsky2024fixed} further extended the classical Mundlak approach and demonstrated that group-level balancing scores sufficiently eliminated between-group confounding effects. However, it crucially relied on the stable unit treatment value assumption (\textit{SUTVA}), which posed significant challenges for its application in network observational studies.

Recent studies have developed inference methods for clustered data settings where interference is confined within clusters. This is sometimes called partial interference \citep{sobel2006randomized}. \citet{hudgens2008toward} proposed a framework for estimating causal effects in two-stage randomized designs under the constraint of partial interference. \citet{liu2019doubly} proposed a doubly robust estimation method under partial interference, which helped reduce bias from model misspecification. \citet{qu2026semiparametric} proposed generalized  augmented inverse propensity weighted (\textit{AIPW}) estimators for heterogeneous partial interference that are semiparametrically efficient and robust to model misspecification.
However, these estimators still failed to adequately capture nonlinear dependencies and higher-order neighborhood effects.
To overcome these limitations, \citet{leung2022graph} introduced a graph neural network (\textit{GNN})-based framework that used message-passing to adaptively aggregate multiorder neighborhood information. In addition, \cite{lin2025scalable} and \citet{hu2025graph} have demonstrated the effectiveness of \textit{GNNs} for estimating individual treatment effects in network settings, addressing computational scalability and confounder disentanglement, respectively. These methods  effectively captured nonlinear and higher-order dependencies, significantly improving causal estimation in complex network environments.

To address the core challenge of violated \textit{SUTVA} due to within-group interdependencies, this paper proposes a generalized Mundlak estimator based on graph neural networks (\textit{GME-GNN}). This estimator synergizes the Mundlak method's strength in modeling between-group heterogeneity with the power of graph neural networks
(\textit{GNNs}) to adjust for network confounding.
It thereby relaxes the three core restrictive assumptions of traditional fixed-effects models discussed earlier.
Specifically, we extend the classical Mundlak approach by constructing generalized group-aggregated balancing statistics that incorporate network-weighted aggregates of treatment and covariate information.
This construction eliminates between-group unobserved confounding while accounting for network dependencies entirely ignored in the original framework.
Meanwhile, we utilize the message-passing mechanism in \textit{GNNs} \citep{corso2020principal} to learn adaptive node embeddings. These embeddings effectively capture dependency structures in the network by iteratively aggregating the features and treatment states of neighboring nodes. For statistical inference, we employ the network \textit{HAC} estimator \citet{kojevnikov2021limit} and adopt the bandwidth selection procedure by \citet{leung2022graph} to account for the machine learning estimation error.

We evaluated the propose theoretical framework through a combination of simulation studies and an empirical analysis examining the impact of the High-Tech Enterprise Certification policy.
We simulate based on the Watts–Strogatz small-world network model, examining scenarios with both low/high inter-group heterogeneity and weak/strong network dependence. The results indicate that the bias of \textit{GME-GNN} is significantly lower than that of the traditional Mundlak estimator under most settings. The empirical analysis is based on a sample of 4693 A‑share listed companies in China in 2024. Firms are grouped into sectors according to the China Securities Regulatory Commission industry classification, and a network is constructed within each group using upstream‑downstream supply‑chain linkages. By using the balanced statistic $\overline{S_g(i)}$ to control for inter-group confounding, the \textit{GME-GNN} estimator captures both the direct causal effect of high‑tech enterprise certification on financial performance and the indirect spillover effects transmitted through the supply‑chain network.

\section{Model Specification}
\label{sec:Model}
\subsection{Set up}\label{subsec1}
The initial fixed-effect model for estimating grouped causal effects is expressed as
\begin{equation}
	Y_i = \alpha_{g(i)} + W_i\tau + \boldsymbol{X}_i^\top \boldsymbol{\beta} + \varepsilon_i,
	\label{eq:tau_1}
\end{equation}
where $g(i)$ denotes the group to which the $i$-th unit belongs. The term $\alpha_{g(i)}$ represents the group fixed effect, capturing the influence of unobserved factors common to all units within the same group. Here, $W_i$ denotes the treatment variable for individual $i$, and $\tau$ represents the average causal effect of the treatment. The vector $\boldsymbol{X}_i$ comprises control variables that may influence the outcome variable $Y_i$, and $\boldsymbol{\beta}$ is the corresponding coefficient vector. The error term $\varepsilon_i$ is assumed to be independent of both $W_i$ and $\boldsymbol{X}_i$.

As \cite{arkhangelsky2024fixed} point out, the ability of equation \eqref{eq:tau_1} to provide valid causal effect estimates is critically dependent on on a number of strong assumptions. This paper elucidates these assumptions by breaking them down into four components. First, the constant treatment effect assumption, this implies homogeneous treatment impacts. In practice, it is likely the effects of the treatment are heterogeneous. Second, the core of equation \eqref{eq:tau_1} is the assumption of group unconfoundedness. Third, the linearity and additivity assumption demands that relationships between the fixed effects $\alpha_{g(i)}$, treatment $W_i$, and covariates $\boldsymbol{X}_i$ are linear and additive. In addition, the independence assumption requires that groups are independent, and that individuals within the same group are also independent.

The original assumption requires that individuals within a group are independent. However, in reality, individuals within a group may be interrelated. Therefore, we first relax this assumption, the model is subsequently expressed in the following form
\[
Y_i = \alpha_{g(i)} + m(\mathbf{W}_g, \mathbf{X}_g, \mathbf{A}_g) + \varepsilon_i,
\]
where $\mathbf{W}_{g}$ denotes the vector of treatment assignments for all units within the group, $\mathbf{X}_{g}$ represents the matrix of observable covariates for all units in the group, and the group comprises $N_{g}$ units in total. The symbol $\mathcal{N}_{{g}}$ refers to a set of units within the group. These units are interconnected through an adjacency matrix $\mathbf{A}_{g}$, where $A_{ij} = 1$ indicates a connection between units $i$ and $j$, whereas $A_{ij} = 0$ signifies no connection.

Due to the presence of group fixed effects, the estimation of causal effects is typically confined to within-group comparisons between treatment and control units. To allow for partial between-group treatment-control comparisons, \cite{mundlak1978pooling} approach can be employed by incorporating group-level covariate averages as additional regressors. \cite{arkhangelsky2024fixed} further demonstrate that these group-average terms not only effectively control between-group heterogeneity but also accommodate more flexible regression specifications. Consequently, the structural model is formally expressed as 
\begin{equation}
	Y_i = \mu_{g}({\mathbf{W}}_g, \mathbf{X}_g, \mathbf{A}_g, \overline{X}_{g(i)}, \overline{W}_{g(i)}, \overline{{A}_{g(i)}{X}_{g(i)}}, \overline{{A}_{g(i)}{W}_{g(i)}},\varepsilon_i),
\end{equation}
where the latter four terms can be written as
$\bar{S}_g = \left( \overline{W}_g, \overline{X}_g, \overline{A_g W_g}, \overline{A_g X_g} \right).$

\subsection{Identification}
We first define the group label $L_g$ that characterizes the group-level latent structural features. 
The group label $L_g$ is a random variable randomly drawn from the group super-population, 
corresponding to the $g$-th group. It fully encodes all inter-group latent heterogeneity that 
distinguishes this group from other groups, and is the native core variable that characterizes 
the group-level latent structural features.

This study establishes the sampling assumption characterizing the group-unit hierarchical data structure.
\begin{Assumption}[Grouped sampling]\label{ass1}
	We randomly sample \( M \) groups from a super-population, where \( g \in \{1, \ldots, M\} \) denotes a generic group. Each group is labeled with its unique group identifier $L_g$. The g-th group (corresponding to $L_g = l$) contains $N_g$ individuals. The network dependencies among individuals within a group are captured by the adjacency matrix \(\boldsymbol{A}_g\).
	The total sample size is given by  \( N = \sum_{g=1}^M N_g \), with each group size \( N_g \geq 2 \) being a random variable.
\end{Assumption}

The sampling framework of this study adopts a two-step sequential random sampling method, following the stratified sampling approach of \cite{arkhangelsky2024fixed}. First, we randomly draw $M$ independent groups from the group super-population, and assign each group a random group label $L_g$ to identify its latent structural features. Second, we proceed to draw $n_g$ individuals from the sub-population of individuals associated with that latent feature $l$. The within-group network adjacency matrix $A_g$ is then constructed based on the actual dependency relationships among these individuals. Under the grouped sampling assumption, the entire dataset can be regarded as \( M \) independent and identically distributed (\textit{i.i.d.}) realizations of the random structure 
\( \left( L_g, \{(W_i, 
X_i, Y_i, A_{ij})\}_{i: g(i) = g} \right) \). This representation ensures statistical independence across groups.

In the field of network causal inference, numerous studies have investigated formulations of unconfoundedness conditions. For instance, \cite{emmenegger2022treatment} and \cite{forastiere2021identification} propose summarizing confounders using low-dimensional network control variables. However, these methods frequently neglect confounding stemming from higher-order neighbor attributes and complex network topologies. In contrast, the group-level unconfoundedness assumption introduced in this study is analogous to the conventional form under the \textit{SUTVA} framework.

\begin{Assumption}[Group-level Unconfoundedness]\label{ass2}
	\[
	\mathbf{W}_g \perp Y_i(\boldsymbol{w}_g) \mid \mathbf{X}_g, \mathbf{A}_g, L_{g(i)}.
	\]
\end{Assumption}
Unlike existing methods, our assumption does not presuppose that the confounding function of \( \mathbf{X}_g \) and \( \mathbf{A}_g \) must take the form of low-dimensional summary statistics. Rather, it permits confounding to be expressed through arbitrarily complex functions of \( \mathbf{X}_g \) and \( \mathbf{A}_g \).

In the presence of network interference within observational data, establishing large-sample theory requires imposing structural restrictions on network dependence to ensure weak dependence. To this end, let \( \mathcal{N}_g(i, s) = \{j \in \mathcal{N}_g: \ell_{\mathbf{A}_g}(i, j) \leq s\} \) denote the \( s \)-th order neighborhood of node \( i \), where the path distance \( \ell_{\mathbf{A}_g}(i, j) \) is the length of the shortest path between nodes \( i \) and \( j \) in the network \( \mathbf{A}_g \). We introduce a nonparametric decaying interference model.
\begin{Assumption}[Approximate Neighborhood Independence]\label{ass3}
	For each group \( g \), there exists a function \( \gamma_{g}: \mathds{R}_+ \to \mathds{R}_+ \) satisfies
	\[
	\sup_{N_g} \{\gamma_{g}(s)\} \to 0 \quad \text{as} \quad s \to \infty,
	\]
	and 
	\begin{align}
		\max_{i:g(i)=g} 
		\mathds{E} \Bigg[ 
		\Bigg| 
		\mu_{g(i)}\bigl( \mathbf{W}_g, \mathbf{X}_g, \mathbf{A}_g, \bar{S}_g(i),\varepsilon_i \bigr) - 
		& \mu_{g(i,s)}\bigl( \mathbf{W}_{\mathcal{N}_g(i,s)}, \mathbf{X}_{\mathcal{N}_g(i,s)}, \mathbf{A}_{\mathcal{N}_g(i,s)}, \bar{S}_g(i), \varepsilon_{i_{\mathcal{N}_g(i,s)}}\bigr)\Bigg| \nonumber \\
		&\,\Bigg|\, 
		\mathbf{W}_g, \mathbf{X}_g, \mathbf{A}_g, \bar{S}_g(i) 
		\Bigg] 
		\leq \gamma_{g}(s) 
		\quad \text{a.s.}
		\label{eq1}
	\end{align}
\end{Assumption}
The assumption adopted in this study is conceptually similar to the approximate neighborhood interference (\textit{ANI}) framework proposed by \cite{leung2022graph}. In contrast, \( \mu_{g(i,s)}(\cdot) \) represents the outcome function for the same unit under a restricted interference model with a neighborhood radius of s.  This restricted function serves as an approximation to the true outcome mapping under limited network interference. The approximation error between the true outcome and the outcome under the \( s \)-neighborhood restricted interference  is governed by \( \gamma_{g}(s) \), which decays to zero as the neighborhood radius \( s \) increases.

For each group \( g \), we define the empirical measure \( \mathds{P}_g \) as a network-weighted distribution of the observed variables \( (W_i, X_i, \{A_{ij}\}_{j:g(j)=g}) \). This weighted formulation explicitly accounts for localized interference effects under the \textit{ANI} framework. Formally, for any set \( B \subseteq \mathds{X} \times \{0,1\} \times \{0,1\}^{N_g} \), \( \mathds{P}_g(B) \) is defined through the following relationship
\begin{equation}
	\mathds{P}_g(B) = \frac{1}{Z_g} \sum_{i: g(i)=g} q_i \cdot \mathbf{1}\{( X_i, W_i, \{A_{ij}\}_{j:g(j)=g}) \in B\}, 
\end{equation}
where the interference weight \( q_i = \sum_{j \in \mathcal{N}_g(i,s)} \frac{1}{\ell_{\mathbf{A}_g}(i,j)^\gamma + \epsilon} \) is defined for each unit \( i \). In this formulation, \( \gamma \geq 1 \) controls the interference decay rate, \( \epsilon > 0 \) ensures numerical stability, and \( Z_g = \sum_i q_i \) serves as the normalizing constant.

\begin{Proposition}[Unconfoundedness with Network-weighted Measure]
	\[
	\mathbf{W}_g \perp Y_i(\boldsymbol{w}_g) \mid \mathbf{X}_g, \mathbf{A}_g, \mathds{P}_{g(i)}.
	\]
\end{Proposition}
Based on the core idea of the empirical measure proposed by \cite{arkhangelsky2024fixed}, this study extends its framework to network structures. This extension also exposes high-dimensional challenges. To make the approach operationally feasible, we replace the empirical distribution with a low-dimensional approximation. For this purpose, we impose structural constraints on the joint distribution and its relation to the covariates using an exponential family representation.
\begin{Assumption}[Exponential Family with Network Weights]\label{ass4}
	Conditional on \( L_g = l \), the distribution of \( (\mathbf{W}_g, \mathbf{X}_g, \mathbf{A}_g) \) belongs to an exponential family with a known sufficient statistic, its conditional density has the following form
	\[
	f_{\mathbf{W}_g, \mathbf{X}_g, \mathbf{A}_g \mid L_g}(\boldsymbol{w}_g, \boldsymbol{x}_g, \boldsymbol{a}_g \mid l) = h(\boldsymbol{w}_g, \boldsymbol{x}_g, \boldsymbol{a}_g) \exp\{\eta(l)^\top S(\boldsymbol{w}_g, \boldsymbol{x}_g, \boldsymbol{a}_g) + \eta_0(l)\}, 
	\]
	with potentially unknown carrier $h(\cdot)$. There exists a function \( \Phi: \mathcal{W} \times \mathcal{X} \times \{0,1\}^{N_g}  \to \mathds{R}^d \), such that the sufficient statistic is \(S(\boldsymbol{w}_g, \boldsymbol{x}_g, \boldsymbol{a}_g) = \sum_{i=1}^{N_g} \Phi\left(w_i, x_i, \{\mathbf{a}_{ij}\}_j, \mathcal{N}_g(i)\right),\) where \(\Phi\left(w_i, x_i, \mathbf{a}_g, \mathcal{N}_g(i)\right) = \Phi_s\left(w_i, x_i, \mathbf{a}_g^{(i,s)}\right).\) 
\end{Assumption}
Drawing on the general theory of exponential families (\cite{wainwright2008graphical}; \cite{barron1991approximation}), we assume that the joint distribution of \( (\mathbf{W}_g, \mathbf{X}_g, \mathbf{A}_g) \) conditional on $L_g$
can be approximated by an exponential family with sufficient statistic $S(\cdot)$. This extends the group-level balancing approach of \cite{arkhangelsky2024fixed} to network settings, where the sufficient statistic incorporates both nodal attributes and neighborhood aggregates.
Although Assumption \ref{ass4} imposes constraints on the joint distribution of \( (\mathbf{W}_g, \mathbf{X}_g, \mathbf{A}_g) \), it does not restrict the distribution of potential outcomes, thereby allowing the treatment effect \( \tau_i \) to be arbitrarily heterogeneous. Based on this, a latent group-level variable \( \bar{U}_g \equiv \eta(L_g) \) can be defined to fully capture the distributional information of \( (\mathbf{W}_g, \mathbf{X}_g, \mathbf{A}_g) \).
\begin{Lemma}\label{le1}
	Under Assumptions 1--4 we have
	\item[(a)]
	\[
	\{\mathbf{W}_g, \mathbf{X}_g, \mathbf{A}_g\} \perp\!\!\!\perp L_g \mid \bar{U}_g,
	\]
	\item[(b)]
	\[
	\mathbf{W}_g \perp\!\!\!\perp Y_i(\boldsymbol{w}_g) \mid \mathbf{X}_g, \mathbf{A}_g, \bar{U}_{g(i)}.
	\]
\end{Lemma}
This lemma shows that we can use the group-level characteristics \( \bar{U}_g \) to replace the group labels \( L_g \), and thus we only need to focus on the group-level characteristics \( \bar{U}_g \). However, it is not possible to ascertain the true value of the group-level parameter \( \bar{U}_g \), nor to obtain a consistent estimate of it \citep{NeymanScott1948}. Drawing on the approach of \cite{arkhangelsky2024fixed}, this paper uses the sufficient statistic \( \overline{S}_g \equiv S(\mathbf{W}_g, \mathbf{X}_g, \mathbf{A}_g)/N_g \) from the exponential family. This statistic serves the role of a balancing statistic analogous to \( \overline{U}_g \).
\begin{Theorem}\label{th1}
	Under Assumptions 1--4 we have
	\[
	\mathbf{W}_g \perp Y_i(\boldsymbol{w}_g) \mid \mathbf{X}_g, \mathbf{A}_g, \overline{S}_{g(i)}.
	\]
\end{Theorem}
Theorem \ref{th1} reduces the high-dimensional objects \( \mathds{P}_{g(i)} \) to a low-dimensional average \( \overline{S}_{g(i)} \) derived from sufficient statistics, enabling cross-group confounding control under low-dimensional conditions.

To define causal effects under network interference, we first introduce the exposure mapping for each individual unit $i$ in group $g$. Formally, an exposure mapping is defined as $T_i = f_g(\mathbf{W}_g, \mathbf{A}_g)$, where $f_g$ denotes a sequence of measurable functions with codomain $\mathcal{T}$, the set of all possible exposure values for units in the group. This mapping converts the global treatment assignment vector $\mathbf{W}_g$ and the group's network adjacency matrix $\mathbf{A}_g$ into a scalar exposure value $T_i$ for unit $i$. It captures the combined effective treatment influence on unit $i$ from its own treatment status and that of its network-connected neighbors. This paper focuses on causal estimands defined through exposure mappings.

Define the generalized propensity score \citep{imbens2000role}
\[
p_t(\boldsymbol{x}_g, \boldsymbol{a}_g, s_i) = \mathds{P}(T_i = t \mid \mathbf{X}_g = \boldsymbol{x}_g, \mathbf{A}_g = \boldsymbol{a}_g, \overline{S}_{g(i)} = s_i).
\]
Depending on the value of the balancing score, the propensity score \( p_t(\boldsymbol{x}_g, \boldsymbol{a}_g, s_i) \) might be equal to zero or one for certain values of \( s \). We cannot expect the overlap assumption \citep{imbens2015causal} to hold. Instead, we are making the following assumption.
\begin{Assumption}[Known overlap]\label{ass5}
	For a known set \( \mathds{B} \) with \( \mathds{P}( \mathbf{X}_g, \mathbf{A}_g, \overline{S}_{g(i)} \in \mathds{B}) > 0 \) there exists \( \eta > 0 \), such that $(a)$ for any \( (\boldsymbol{x}_g, \boldsymbol{a}_g, s_i) \in \mathds{B} \) we have \( \eta < p_t(\boldsymbol{x}_g, \boldsymbol{a}_g, s_i) < 1 - \eta \) and $(b)$ for \( (\boldsymbol{x}_g, \boldsymbol{a}_g, s_i) \notin \mathds{B} \), \( p_t(\boldsymbol{x}_g, \boldsymbol{a}_g, s_i) \in \{0, 1\} \).
\end{Assumption}
The purpose of set $\mathds{B}$ is to select valid samples with non-extreme propensity scores, which ensures that the causal effect estimates are reliable. The set $\mathds{B}$ in the method is uniquely and stably identified empirically through a thresholding method based on GNN-treated propensity scores, where the threshold $\eta$ is determined by the characteristics of the actual data. This assumption has two key aspects. First, the propensity score $p_t(\boldsymbol{x}_g, \boldsymbol{a}_g, s_i)$ is well-behaved and valid only for samples that fall within this set $\mathds{B}$. Second, the assumption requires that researchers can explicitly identify the specific boundaries of $\mathds{B}$ from the observed data using the method described \citep{arkhangelsky2024fixed,crump2009dealing}.

We now define the target estimands and focus on two of them. The first estimand targets the average treatment effect (\textit{ATE}) for a subpopulation within the overall population that meets specific conditions, expressed as
\begin{equation}
	\tau_{(t,t')} \equiv \mathds{E}\left[ Y_i(t) \mid  (\mathbf{X}_g, \mathbf{A}_g, \overline{S}_{g(i)}) \in \mathds{B} \right] - \mathds{E}\left[ Y_i(t') \mid  (\mathbf{X}_g, \mathbf{A}_g, \overline{S}_{g(i)}) \in \mathds{B} \right].\label{eq3}
\end{equation}
The second estimand is a sample-weighted version of equation \eqref{eq3}. We first define the conditional treatment effect for a single unit in the sample as
\[
\tilde{\tau}_{i(t,t')} \equiv \mathds{E}\left[ Y_i(t) \mid \mathbf{X}_g, \mathbf{A}_g, \overline{S}_{g(i)} \right] - \mathds{E}\left[ Y_i(t')  \mid  \mathbf{X}_g, \mathbf{A}_g, \overline{S}_{g(i)} \right].
\]
Then, for the units in the sample that satisfy \( (\mathbf{X}_g, \mathbf{A}_g, \overline{S}_{g(i)}) \in \mathds{B} \), a weighted average of \( \tilde{\tau}_{i(t,t')} \) is taken using the number of such units in the sample as the weight, defined as
\[
\tilde{\tau}_{(t,t')} \equiv \frac{\sum_{i=1}^N \mathbf{1}_{(\mathbf{X}_g, \mathbf{A}_g, \overline{S}_{g(i)}) \in \mathds{B}} \tilde{\tau}_{i(t,t')}}{\sum_{i=1}^N \mathbf{1}_{(\mathbf{X}_g, \mathbf{A}_g, \overline{S}_{g(i)}) \in \mathds{B}}}.
\]
\begin{Corollary}
	Suppose Assumptions \ref{ass1}--\ref{ass5} hold. Then \( \tau_{(t,t')} \) is identified.
\end{Corollary}
\begin{proof}
	The results follow from the identity
	\[
	\begin{aligned}
		\tau_{(t,t')}
		&= \mathds{E}\left[ Y_i(t) \mid  (\mathbf{X}_g, \mathbf{A}_g, \overline{S}_{g(i)}) \in \mathds{B} \right] - \mathds{E}\left[ Y_i(t') \mid  (\mathbf{X}_g, \mathbf{A}_g, \overline{S}_{g(i)}) \in \mathds{B} \right]\\
		&= \mathds{E}\left[ (\frac{\mathbf{1}\{T_i = t\}}{p_t(\mathbf{X}_g, \mathbf{A}_g, \overline{S}_{g(i)})} - \frac{\mathbf{1}\{T_i = t'\}}{p_{t'}(\mathbf{X}_g, \mathbf{A}_g, \overline{S}_{g(i)})}) Y_i \bigg| (\mathbf{X}_g, \mathbf{A}_g, \overline{S}_{g(i)}) \in \mathds{B} \right].
	\end{aligned}
	\]
\end{proof}
\subsection{Estimation}\label{subsec3}
Define the generalized propensity score and outcome regression, respectively, as
\begin{align}
	p_t(\mathbf{X}_g, \mathbf{A}_g, \overline{S}_{g(i)}) &= \mathds{P}(T_i = t \mid \mathbf{X}_g, \mathbf{A}_g, \overline{S}_{g(i)}), \notag \\
	\mu_t(\mathbf{X}_g, \mathbf{A}_g, \overline{S}_{g(i)}) &= \mathds{E}[Y_i \mid T_i = t, \mathbf{X}_g, \mathbf{A}_g, \overline{S}_{g(i)}]. \tag{7}
\end{align}
We use \( \hat{p}_t(\mathbf{X}_g, \mathbf{A}_g, \overline{S}_{g(i)}) \) and \( \hat{\mu}_t(\mathbf{X}_g, \mathbf{A}_g, \overline{S}_{g(i)}) \) as the corresponding \textit{GNN} estimators, as described in section \ref{subsec4}. Define indicator variable \( B_i \equiv \mathbf{1}_{(\mathbf{X}_g, \mathbf{A}_g, \overline{S}_{g(i)}) \in \mathds{B}} \) and let \( \bar{B} \equiv \frac{1}{M} \sum_{g=1}^M \frac{1}{N_g} \sum_{i: g(i)=g} B_i \) be the estimate of the share of units for which we have overlap. A standard doubly robust estimator for multi-valued treatments takes the form
\begin{align}
	\hat{\tau}_i(t, t') &= \frac{\mathbf{1}\{T_i = t\}(Y_i - \hat{\mu}_t(\mathbf{X}_g, \mathbf{A}_g, \overline{S}_{g(i)}))}{\hat{p}_t(\mathbf{X}_g, \mathbf{A}_g, \overline{S}_{g(i)})} + \hat{\mu}_t(\mathbf{X}_g, \mathbf{A}_g, \overline{S}_{g(i)}) \notag \\
	& \quad - \frac{\mathbf{1}\{T_i = t'\}(Y_i - \hat{\mu}_{t'}(\mathbf{X}_g, \mathbf{A}_g, \overline{S}_{g(i)}))}{\hat{p}_{t'}(\mathbf{X}_g, \mathbf{A}_g, \overline{S}_{g(i)})} - \hat{\mu}_{t'}(\mathbf{X}_g, \mathbf{A}_g, \overline{S}_{g(i)}), 
\end{align}
and its group-level aggregate
\[
\hat{\tau}_g(t, t') = \frac{1}{N_g} \sum_{i: g(i)=g} B_i \hat{\tau}_i(t, t'). 
\]
We define our proposed \textit{GME\text{-}GNN}
\[
\hat{\tau}_{\textit{GME\text{-}GNN}}(t, t') = \frac{1}{M} \sum_{g=1}^M \hat{\tau}_g(t, t'). 
\]

This integrated framework is named the \textit{GME\text{-}GNN} estimator, where \textit{GNN} specifically refers to the implemented twisted \textit{GNN} to address network confounding and nonlinear dependencies. In the context of network data, the message-passing mechanism of \textit{GNN} can efficiently capture dependencies between nodes and higher-order network effects. Therefore, it is adopted in this paper as the specific implementation approach (see Section \ref{subsec4} for details).

For asymptotic variance estimation, we employ the network \textit{HAC} estimator by \cite{leung2022graph}, extending the framework of \cite{kojevnikov2021limit}.
\begin{align}
	\hat{\sigma}^2 = \frac{1}{M} \sum_{g=1}^M \frac{1}{N_g} \sum_{i \in \mathcal{N}_g} &\sum_{j \in \mathcal{N}_g} B_i \left( \hat{\tau}_i(t, t') - \hat{\tau}_{\textit{GME\text{-}GNN}}(t, t') \right) \notag\\
	& \times B_j \left( \hat{\tau}_j(t, t') - \hat{\tau}_{\textit{GME\text{-}GNN}}(t, t') \right) \mathbf{1} \left\{ \ell_{\mathcal{A}}(i, j) \leqslant b_{g} \right\}. 
\end{align}
This estimator computes covariances only for node pairs within a specific bandwidth \( b_{n_g} \), effectively discounting dependencies between distant nodes. The bandwidth \( b_{n_g} \) thus serves as a crucial tuning parameter that governs the bias-variance trade-off.
\[
b_{g} = \left\lceil \tilde{b}_{g} \right\rceil, \quad \text{for } \tilde{b}_{g} =
\begin{cases}
	\frac{1}{4} \mathcal{L}(\mathbf{A}_g), & \text{if } \mathcal{L}(\mathbf{A}_g) < 2 \cdot \frac{\log N_g}{\log \delta(\mathbf{A}_g)} ,\\
	\mathcal{L}(\mathbf{A}_g)^{1/4}, & \text{otherwise},
\end{cases} 
\]
where \( \lceil \cdot \rceil \) rounds up to the nearest integer, \( \delta(\mathbf{A}_g) = N_g^{-1} \sum_{i,j} A_{ij} \) is the average degree, and \( \mathcal{L}(\mathbf{A}_g) \) is the average path length. This is similar to the proposal of \cite{leung2022graph}.
\subsection{Construction of Sufficient Statistics Based on Exponential Families}\label{subsec5}

Theorem \ref{th1} shows that the group-level balancing statistic $\overline{S}_g = S(\mathbf{W}_g, \mathbf{X}_g, \mathbf{A}_g)/N_g$ is sufficient to control for group-level confounding. To operationalize this result, this section provides a specific construction of $\overline{S}_g$ based on the theory of exponential families \citep{wainwright2008graphical}. This construction operationalizes Assumption \ref{ass4} by specifying a concrete form of $\bar{\Phi}(\cdot)$  that remains low-dimensional while encoding both nodal attributes and first-order neighborhood aggregates. Higher-order dependencies are captured through the \textit{GNN} estimation framework in Section \ref{subsec4}, rather than through the sufficient statistic itself.

First, we divide the individual-level information within a group into self-information and network-dependent information.
The self-information includes the treatment status $\Phi'(w_i) = w_i$ and the covariate $\Phi''(x_i) = x_i$.
The network-dependent information includes the neighboring treatment status
\[
\Phi_{(s)}^{'''}(w_i,\{a_{ij}\}_j,\mathcal{N}_g(i)) = \sum_{j \in \mathcal{N}_g(i,s)} a_{ij}^{(s)} w_j,
\]
and the neighboring covariate
\[
\Phi''''(x_i,\{a_{ij}\}_j,\mathcal{N}_g(i)) = \sum_{j \in \mathcal{N}_g(i)} a_{ij} x_j.
\]
We construct the local statistical vector $\Phi(\cdot)$ for individual $i$ to satisfy $\Phi(w_i, x_i, \mathbf{a}_g, \mathcal{N}_g(i)) = \Phi_s(w_i, x_i, \mathbf{a}_g^{(i,s)})$. Its explicit form is 
\[
\Phi(w_i,x_i,\{a_{ij}\}_j,\mathcal{N}_g(i)) =
\begin{bmatrix}
	\Phi'(w_i) \\[2pt]
	\Phi''(x_i) \\[2pt]
	\Phi'''(w_i, \{a_{ij}\}_j, \mathcal{N}_{g(i)}) \\[2pt]
	\Phi''''(x_i, \{a_{ij}\}_j, \mathcal{N}_{g(i)})
\end{bmatrix}
=
\begin{bmatrix}
	w_i \\[2pt]
	x_i \\[2pt]
	\sum_{j \in \mathcal{N}_g(i)} a_{ij} w_j \\[2pt]
	\sum_{j \in \mathcal{N}_g(i)} a_{ij} x_j
\end{bmatrix}.
\]
Based on the additive structure of the global sufficient statistic as the sum of individual local statistics ($S = \sum_{i=1}^{N_g} \Phi(\cdot)$) in the assumption, we sum the local statistic vectors of all individuals by component dimensions to obtain the global sufficient statistic
\[
S(\mathbf{W}_g,\mathbf{X}_g,\mathbf{A}_g) = \sum_{i:g(i)=g} \Phi(w_i,x_i,\{a_{ij}\}_j,\mathcal{N}_g(i))
= \begin{bmatrix}
	\sum_{i=1}^{N_g} w_i \\[4pt]
	\sum_{i=1}^{N_g} x_i \\[4pt]
	\sum_{i=1}^{N_g} \sum_{j \in \mathcal{N}_g(i)} a_{ij} w_j \\[4pt]
	\sum_{i=1}^{N_g} \sum_{j \in \mathcal{N}_g(i)} a_{ij} x_j
\end{bmatrix}.
\]
According to Theorem~\ref{th1}, we divide the global sufficient statistic by the group size $N_g$ to obtain the balancing statistic
\(\overline{S}_g(i) = \frac{1}{N_g} S(\mathbf{W}_g,\mathbf{X}_g,\mathbf{A}_g) = \begin {bmatrix}
\overline {W}_g,
\overline {X}_g,
\overline {A_g W_g},
\overline {A_g X_g}
\end {bmatrix}^\top\).
	This group-mean form ensures the comparability of statistics across different group sizes and directly satisfies the conditions in Theorem~\ref{th1}.
	
	\subsection{Graph Neural Networks}\label{subsec4}
	The estimator proposed in this study is based on a doubly robust causal inference framework, which is not inherently tied to any specific model choice. Theoretically, the estimators for the propensity score  \( \hat{p}_t(\mathbf{X}_g, \mathbf{A}_g, \overline{S}_{g(i)}) \) and the outcome regression \( \hat{\mu}_t(\mathbf{X}_g, \mathbf{A}_g, \overline{S}_{g(i)}) \)  can be replaced with other parametric or nonparametric models. In network-based causal inference, accurately estimating these two components requires capturing group-level heterogeneity, dependencies in network topology, and higher-order neighborhood effects simultaneously. However, traditional estimation methods often struggle to adequately meet these demands. Specifically, traditional parametric models typically assume a linear relationship between covariates and the treatment/outcome, which cannot adapt to the complex nonlinear dependencies present in network data. While general-purpose machine learning models can handle nonlinear problems, they lack inherent compatibility with network structure. This makes it difficult for them to characterize the dependencies between nodes and higher-order network effects. Furthermore, existing methods for network control often rely on low-dimensional summary statistics to account for confounding. This approach overlooks the high-dimensional confounding problems introduced by higher-order neighborhood attributes and complex network topology, resulting in insufficient control of confounding factors.
	
	Theoretically, if other models are to be used for estimating the propensity score and outcome regression, they must satisfy the three core requirements we have defined for the network confounding scenario. First, the model must possess permutation invariance. This means that for nodes with identical positions and connection structures in the network, the model's predictions should remain consistent regardless of any changes to their labels or identifiers. Second, the model must adapt to the local characteristics of network confounding. It should effectively capture higher order network confounding effects through a limited neighborhood around each node, thereby satisfying network approximate sparsity. Third, the model must achieve a sufficient convergence rate. It needs to meet the \textit{GNN} convergence rate standard set in Assumption \ref{ass7} of the paper. This is necessary to provide theoretical support for the asymptotic normality and variance consistency of the estimation results. 
	
	In contrast, Graph Neural Networks (\textit{GNNs}) leverage a message-passing mechanism that adaptively aggregates information across layers through iterative updates. This enables them to efficiently capture nonlinear dependencies and higher-order interactive effects within network structures, making \textit{GNNs} highly suitable for the research scenario in this paper. Moreover, \textit{GNNs} satisfy all three requirements outlined above \cite{leung2022graph}. Therefore, this paper adopts \textit{GNNs} as the estimation framework for both the generalized propensity score and the outcome regression.
	
	To implement the estimation framework, we employ \textit{GNNs}, which utilize a message-passing mechanism to update node embeddings through layer-wise aggregation of neighborhood information.
	The embedding of a node at layer \( l \) is constituted by its historical state and neighborhood states. Specifically, the embedding \(h_i^{(l)}\) of node \( i \) is generated from two input sources
	\[h_i^{(l)} = \Phi_{0l}\left(h_i^{(l-1)}, \Phi_{1l}\left(h_i^{(l-1)}, \{h_j^{(l-1)}: A_{ij} = 1\}\right)\right),\]
	where \( \Phi_{0l}(\cdot) \) and \( \Phi_{1l}(\cdot) \) are parameterized vector-valued functions. We initialize the process with \( h_i^{(0)} = \mathbf{X}_i \) at the input layer, meaning that the initial node embedding incorporates no network information. The depth \( L \) of the \( L \)-neighborhood of the unit is used for prediction.
	
	Let \( \mathcal{F}_{\text{GNN}}(L) \) denote all \( L \)-layer GNNs with arbitrary \( \Phi_{0l}, \Phi_{1l} \). For any \( f \in \mathcal{F}_{\text{GNN}}(L) \), the output \( f(\mathbf{X}_g, \mathbf{A}_g, \overline{S}_{g(i)}) \) corresponds to \( \mathbf{h}_i^{(L)} \). The GNN estimator minimizes empirical risk
	\[
	\hat{f}_{\text{GNN}} \in \arg\min_{f \in \mathcal{F}_{\text{GNN}}(L)} \sum_{i:g(i)=g} \ell(Y_i, f(\mathbf{X}_g, \mathbf{A}_g, \overline{S}_{g(i)})).
	\]
	
	Loss functions are specified by outcome type. Continuous outcomes (\( Y_i \in \mathds{R} \)) use squared loss \( \ell = \frac{1}{2}(Y_i - f(\cdot))^2 \), outputting conditional expectation
	\[
	f^*(\mathbf{X}_g, \mathbf{A}_g, \overline{S}_{g(i)}) = \left(\mathds{E}[Y_i \mid \mathbf{X}_g, \mathbf{A}_g, \overline{S}_{g(i)}]\right)_{i:g(i)=g}^{N_g}.
	\]
	Binary outcomes (\( Y_i \in \{0,1\} \)) adopt logistic loss \( \ell = -Y_i f(\cdot) + \log(1 + \exp(f(\cdot))) \), estimating log-odds
	\[
	f^*(\mathbf{X}_g, \mathbf{A}_g, \overline{S}_{g(i)}) = \left(\log\left(\frac{\mathds{E}[Y_i \mid \mathbf{X}_g, \mathbf{A}_g, \overline{S}_{g(i)}]}{1 - \mathds{E}[Y_i \mid \mathbf{X}_g, \mathbf{A}_g, \overline{S}_{g(i)}]}\right)\right)_{i:g(i)=g}^{N_g}.
	\]
	
	To implement the doubly robust estimator from Section \ref{subsec3}, we estimate the propensity scores by substituting \(Y_i\) with \(\mathbf{1}\{T_i = t\}\) and minimizing the logistic loss, resulting in
	\[
	\hat{p}_t(\mathbf{X}_g, \mathbf{A}_g, \overline{S}_{g(i)}) = \frac{\exp(\hat{f}_{\text{GNN},p}(\mathbf{X}_g, \mathbf{A}_g, \overline{S}_{g(i)}))}{1 + \exp(\hat{f}_{\text{GNN},p}(\mathbf{X}_g, \mathbf{A}_g, \overline{S}_{g(i)}))}.
	\]
	For outcome regression with \(\mathds{R}\)-valued outcomes, we apply the squared-error loss exclusively to units where \(T_i = t\), obtaining the estimator
	\[
	\hat{\mu}_t(\mathbf{X}_g, \mathbf{A}_g, \overline{S}_{g(i)}) = \hat{f}_{\text{GNN},\mu}(\mathbf{X}_g, \mathbf{A}_g, \overline{S}_{g(i)}).
	\]
	
	Both $\hat{f}_{\text{GNN},\mu}$ and $\hat{f}_{\text{GNN},p}$ belong to the \text{GNN} function class $\mathcal{F}_{\text{GNN}}(L)$ defined in this section, but they are trained separately with different loss functions and independent parameters. For notational simplicity in the subsequent sections, we collectively refer to them as $\hat{f}_{\text{GNN}}$ when the context is clear.
	\section{Asymptotic Theory}\label{Theory}
	This section establishes the asymptotic properties of the \textit{GME\text{-}GNN} estimator. We consider an asymptotic sequence where the number of groups \( M \to \infty \) and the number of units within each group \( N_g \to \infty \) . Under this framework, the model primitives \( (f_{g}, \mu_{g}) \) and network structure \( \mathbf{A}_g \) are allowed to exhibit heterogeneity across different groups. Define the influence function for each unit as
	\begin{align}
		\varphi_{t,t'}(i) =& \frac{\mathbf{1}\{T_i = t\}(Y_i - \mu_t(\mathbf{X}_g, \mathbf{A}_g, \overline{S}_{g(i)}))}{p_t(\mathbf{X}_g, \mathbf{A}_g, \overline{S}_{g(i)})} + \mu_t(\mathbf{X}_g, \mathbf{A}_g, \overline{S}_{g(i)}) \notag\\
		& - \frac{\mathbf{1}\{T_i = t'\}(Y_i - \mu_{t'}(\mathbf{X}_g, \mathbf{A}_g, \overline{S}_{g(i)}))}{p_{t'}(\mathbf{X}_g, \mathbf{A}_g, \overline{S}_{g(i)})} - \mu_{t'}(\mathbf{X}_g, \mathbf{A}_g, \overline{S}_{g(i)}) - \tau(t, t'). 
	\end{align}
	The group-average influence function is defined as \( \varphi_{t,t'}(g) = \frac{1}{N_g} \sum_{i: g(i)=g} B_i \varphi_{t,t'}(i) \). The asymptotic variance is given by
	\[
	\sigma^2 = \frac{1}{M} \frac{1}{\bar{B}^2} \sum_{g=1}^M \text{Var} \left( \frac{1}{\sqrt{N_g}} \sum_{i: g(i)=g} B_i \varphi_{t,t'}(i) \,\bigg|\, \mathbf{X}_g, \mathbf{A}_g, \overline{S}_{g(i)} \right).
	\]
	\begin{Assumption}[Moments]\label{ass6}
		\( (a) \) There exists \( C < \infty \) and \( p > 4 \) such that for an arbitrary individual \( i \) in any group, and \( \boldsymbol{w}_g \in \{0,1\}^{N_g} \), \( \mathds{E}[|Y_i(\boldsymbol{w}_g)|^p \mid \mathbf{X}_g, \mathbf{A}_g, \overline{S}_{g(i)}] < C \) a.s. 
		\( (b) \) \( \frac{\eta}{2} < \hat{p}_t(\mathbf{X}_g, \mathbf{A}_g, \overline{S}_{g(i)}) < 1 - \frac{\eta}{2} \) on \( \mathds{B} \).
		\( (c) \)$\liminf_{\substack{M \to \infty \\ N_g \to \infty\ \text{for all } g}} \sigma^2 > 0\ \text{a.s.}$
	\end{Assumption}
	Part \( (a) \) uniform boundedness of the \( p \)-th moment for potential outcomes establishes the foundation for central limit theorem (\textit{CLT}). Part \( (b) \) requires sufficient overlap for the propensity scores. Part \( (c) \) is a standard non-degeneracy condition.
	\begin{Assumption}[GNN Rates]\label{ass7}
		For any \( t \in \mathcal{T} \), 
		\( \frac{1}{M} \sum_{g=1}^M \frac{1}{N_g} \sum_{i: g(i)=g} ( \hat{p}_t(\mathbf{X}_g, \mathbf{A}_g, \overline{S}_{g(i)}) \\  
		- p_t(\mathbf{X}_g, \mathbf{A}_g, \overline{S}_{g(i)}) )^2 \)
		and 
		\( \frac{1}{M} \sum_{g=1}^M \frac{1}{N_g} \sum_{i: g(i)=g} \left( \hat{\mu}_t(\mathbf{X}_g, \mathbf{A}_g, \overline{S}_{g(i)}) - \mu_t(\mathbf{X}_g, \mathbf{A}_g, \overline{S}_{g(i)}) \right)^2 \)
		are \( o_p(1) \), and their product is \( o_p\left( \frac{1}{MN_g} \right) \).
	\end{Assumption}
	This assumption represents a standard high-level condition for machine learning estimators in causal inference \citep{farrell2015robust}. This condition does not prescribe specific \textit{GNN} architectures but directly constrains the asymptotic properties of their output statistics.
	
	\( \psi \)-dependence is a fundamental property of weakly dependent sequences, which quantifies correlations in networked data through the decay rate of covariances \citep{kojevnikov2021limit}. This property serves as a prerequisite for the validity of the \textit{CLT} in the presence of network interference. Building upon the \textit{GNN} framework established by \cite{leung2022causal}, this assumption explicitly constrains the relationship between the dependence decay rate \( \psi_{g}(s) \) and the growth of neighborhoods \( |\mathcal{N}_{g}(i, s)| \). Define
	\begin{align*}
		\mathcal{N}_{g}^\partial(i, s) = \{ j \in \mathcal{N}_{g}: \ell(i, j) = s \},  
		\delta_{g}^\partial(s; k) = \frac{1}{N_g} \sum_{i: g(i)=g} \left| \mathcal{N}_{g}^\partial(i, s) \right|^k. 
	\end{align*}
	They respectively represent the set of nodes within the group that are at a shortest-path distance of exactly \( s \) from node \( i \), and the \( k \)-th order moment of this set of nodes. Let
	\[
	\Delta_{g}(s, r; k) = \frac{1}{N_g} \sum_{i:g(i)=g} \max_{j \in \mathcal{N}_{g}^\partial(i, s)} \left| \mathcal{N}_{g} (i, r) \setminus \mathcal{N}_{g}(j, s - 1) \right|^k,
	\]
	\[
	c_{g}(s, r; k) = \inf_{\alpha > 1} \Delta_{g}(s, r; k\alpha)^{1/\alpha} \delta_{g}^\partial\left(s; \alpha/(\alpha - 1)\right)^{1 - 1/\alpha},
	\]
	and
	\begin{equation}
		\psi_{n_g}(s) = \max_{i \in \mathcal{N}_g} (\gamma_{g}(s/4)+\eta_{g}(s/4)). \label{phs}
	\end{equation}
	
	The first equation characterizes the local network structural variation within the group. The second equation measures the network density within the group. The third quantity \( \psi_{n_g}(s) \) establishes a theoretical upper bound for covariance between influence functions. It constraints \( \varphi_{t,t'}(i) \) and \( \varphi_{t,t'}(j) \) for any node pair within path distance \( s \) in network \( \mathbf{A}_g \). Define
	\[T_i^{(s)} \equiv f_{g,s}\left( \mathbf{W}_{\mathcal{N}_g(i,s)}, \mathbf{A}_{\mathcal{N}_g(i,s)} \right).\]

	\begin{Assumption}[Weak Dependence]\label{ass8}
		\( (a) \) \( \{\varepsilon_i\}_{i:g(i)=g}^{N_g} \) is independently distributed conditional on \( (\mathbf{X}_g, \mathbf{A}_g, \overline{S}_{g(i)}) \).
		
		\( (b) \)  For each group \( g \) ,\( s \geq 0 \), there exists a function \( \eta_{g}: \mathds{R}_+ \to \mathds{R}_+ \) satisfies
		\(\sup_{N_g} \{\eta_{g}(s)\} \to 0 \quad \text{as} \quad s \to \infty,\) and
		\[\max_{i: g(i) = g} \mathds{E}\left[ \left| \mathbf{1}\{T_i = t\} - \mathbf{1}\{T_i^{(s)} = t\} \right| \middle| X_g, A_g, \overline{S}_{g(i)} \right] \leq \eta_g(s) \quad \text{a.s.}\]
		
		\( (c) \) \( \sup_{N_g} \max_{s \geq 1} \psi_{n_g}(s) < \infty \) a.s.
		\( (d) \) For \( p \) in Assumption \ref{ass6}\( (a) \), some positive sequence \( v_{g} \to \infty \) and any \( k \in \{1,2\} \),
		\[
		\frac{1}{{N_g}^{k/2}} \sum_{s=0}^{\infty} c_{g}(s,v_{g};k)\psi_{n_g}(s)^{1-(2+k)/p} \to 0, \quad {N_g}^{3/2}\psi_{n_g}(v_{g})^{1-1/p} \to 0, \quad \text{and}
		\]
		\[
		\limsup_{N_g \to \infty} \sum_{s=0}^{\infty} \sigma_{g}^2(s;2)^{1/2} \gamma_{g}(s/2)^{1-2/p} < \infty \quad \text{a.s.} 
		\]
	\end{Assumption}
	Parts \( (a) \) and Parts \( (b) \) is employed to prove that the sequence set \(\{\phi_{\{t,t'\}}(i)\}_{i:g(i)=g}\) is  \(\psi\)-dependent. 
	Perhaps the central requirement lies in part \( (d) \), which characterizes the asymptotic behavior of the three variables. The first two of these correspond to the \textit{ND} condition proposed by \cite{kojevnikov2021limit}, which they employed to establish a \textit{CLT}. The third variable is similar in nature and is used for the asymptotic linearization of the doubly robust estimator under network dependence. The proof of this expression can be found in appendix \( B \).1 of \cite{leung2022causal}.
	
	\begin{Lemma}[]
		Under Assumptions \ref{ass1}, \ref{ass2}, \ref{ass3}, \ref{ass6}(a) and (b),  \ref{ass8}(a) and (b), for any $t,t' \in \mathcal{T}$, $\{\varphi_{t,t'}(i)\}_{i:g(i)=g}^{n_g}$ is conditionally $\psi$-dependent given $(\mathbf{X}_g, \mathbf{A}_g, \overline{S}_{g(i)})$  with dependence coefficient $\psi_{n_g}(s)$ defined in equation \ref{phs}.
	\end{Lemma}
	
	\begin{Theorem}\label{th2}
		Under Assumptions 1--8,
		\[
		\sigma^{-1/2} \sqrt{MN_g} \left( \hat{\tau}_{\textit{GME-GNN}}(t, t') - \tilde{\tau}_{(t,t')} \right) \stackrel{d}{\to} \mathcal{N}(0, 1).
		\]
	\end{Theorem}
	Critically, because conditioning is performed on the observed network structure and covariates, the estimator \( \hat{\sigma}^2 \) is not consistent. We prove this estimator exhibits asymptotic conservatism, following \cite{leung2022graph}. Consequently, it systematically overestimates the true variance in large samples. Define
	\[
	\mathcal{J}_g(s, r) = \left\{ (i, j, k, l) \in \mathcal{N}_{g}^4 : k \in \mathcal{N}_{g}(i, r), l \in \mathcal{N}_{g}(j, r), \ell_{\mathbf{A}_g}(i, j) = s \right\}.
	\]
	
	\begin{Assumption}[\textit{HAC}]\label{ass9} 
		\( (a) \) For some $C > 0$, $i \in \mathcal{N}_{g}$, and $t \in \mathcal{T}$, $\left| \max\{Y_i, \hat{\mu}_t(\boldsymbol{X}_g, \boldsymbol{A}_g,\overline{S}_{g(i)})\} \right| < C$ a.s.
		
		\( (b) \) $\frac{1}{M} \sum_{g=1}^{M}\frac{1}{N_g} \sum_{i:g(i)=g} (\hat{p}_t( \boldsymbol{X}_g, \boldsymbol{A}_g,\overline{S}_{g(i)}) - p_t( \boldsymbol{X}_g, \boldsymbol{A}_g,\overline{S}_{g(i)}))^2$ and
		$\frac{1}{M} \sum_{g=1}^{M}\frac{1}{N_g} \sum_{i:g(i)=g}\\ (\hat{\mu}_t(\boldsymbol{X}_g, \boldsymbol{A}_g,\overline{S}_{g(i)}) - \mu_t( \boldsymbol{X}_g, \boldsymbol{A}_g,\overline{S}_{g(i)}))^2$ are $o_P({(MN_g)}^{-1/2}).$
		
		\( (c) \) For some $\epsilon \in (0, 1)$ and $b_{g} \to \infty$, $\lim_{{N_g} \to \infty} {N_g}^{-1} \sum_{s=0}^{\infty} c_{g}(s, b_{g}; 2) 
		\psi_{g}(s)^{1 - \epsilon} = 0$ a.s. 
		
		\( (d) \) $\frac{1}{M} \sum_{g=1}^{M}\frac{1}{N_g} \sum_{i:g(i)=g} N_g(i, b_{g}) = o_p(\sqrt{MN_g})$. 
		
		\( (e) \) $\frac{1}{M} \sum_{g=1}^{M}\frac{1}{N_g} \sum_{i:g(i)=g} N_g(i, b_{g})^2 = O_p(\sqrt{MN_g})$. 
		
		\( (f) \) $\frac{1}{M} \sum_{g=1}^{M}\sum_{s=0}^{\infty} \left| \mathcal{J}_{g}(s, b_{g}) \right| \psi_{g}(s) = o((MN_g)^2)$.
	\end{Assumption}
	Part \( (a) \) of Assumption \ref{ass9} strengthens the corresponding condition in Assumption \ref{ass6}\( (a) \). Assumption \ref{ass9}\( (b) \) mildly strengthens Assumption \ref{ass7}, with the additional requirement that these two mean squared errors converge at \( o_P({(MN_g)}^{-1/2}) \). Assumption \ref{ass9}\( (c) \) is consistent with Assumption 4.1(iii) of \cite{kojevnikov2021limit}. Assumptions \ref{ass9}\( (d) \)\text{-}\( (f) \) correspond to Assumptions 7\( (b) \)\text{-}\( (d) \) in \cite{leung2022graph}, which characterize the bias properties of the variance estimator.
	\begin{Theorem}\label{th3}
		Define $\tilde{\varphi}_{t,t'}(i)$ by replacing $\tau(t,t')$ in the definition of $\varphi_{t,t'}(i)$ with $\tau_i(t,t') = \mathds{E}[Y_i \mid T_i = t, \boldsymbol{X}_g, \boldsymbol{A}_g,\overline{S}_{g(i)}] - \mathds{E}[Y_i \mid T_i = t',  \boldsymbol{X}_g, \boldsymbol{A}_g,\overline{S}_{g(i)}]$. Let
		\[
		\hat{\sigma}_{*}^{2} = \frac{1}{M} \sum_{g=1}^{M}\frac{1}{N_g} \sum_{i:g(i)=g}\sum_{j:g(j)=g} B_i\tilde{\varphi}_{t,t'}(i)B_j \tilde{\varphi}_{t,t'}(j) \mathbf{1}\{\ell_{\boldsymbol{A}_g}(i, j) \leq b_{g}\}
		\]
		and
		\begin{align*}
			R = \frac{1}{M} \sum_{g=1}^{M}\frac{1}{N_g} \sum_{i:g(i)=g}\sum_{j:g(j)=g} B_i \left( \tau_i(t,t') - \tau(t,t') \right) B_j\left( \tau_j(t,t') - \tau(t,t') \right)\\
			\times \mathbf{1}\{\ell_{\boldsymbol{A}_g}(i, j) \leq b_{g}\}.
		\end{align*}
		Under Assumption \ref{ass9} and the assumptions of Theorem \ref*{th2},
		\[
		\hat{\sigma}^2 = \hat{\sigma}_{*}^{2} + R + o_p(1) \quad \text{and} \quad \left| \hat{\sigma}_{*}^{2} - \sigma^2 \right| \stackrel{p}{\to} 0.
		\]
	\end{Theorem}
	The variance estimator \( \hat{\sigma}^2 \) exhibits asymptotic bias, which is captured by the term \( R \). This result is established through Theorem 2 of \cite{leung2022graph} and Proposition 4.1 of \cite{kojevnikov2021limit}.
	
	For variance estimation \( \hat{\sigma}^2 \) in this context, the \textit{HAC} estimator is particularly appropriate. It operates by assigning weights to pairs of units through a kernel function based on their network distance, thereby capturing the dependence structure. Following the approach of \cite{leung2022causal}, the \textit{HAC} variance estimator \( \hat{\sigma}^2 \) in this paper employs the uniform kernel, defined as \( \mathbf{1}\{\ell_{A_g}(i, j) \leq b_{g}\} \).
	
	\section{Simulation study}
	We conduct simulation experiments to evaluate the finite-sample performance of the proposed \textit{GME\text{-}GNN}estimator under two key dimensions: (i) the strength of network dependence, and (ii) the degree of group-level heterogeneity. These two dimensions are central to the theoretical framework developed in Sections \ref{sec:Model} and \ref{Theory}.
	\subsection{Design}
	
	We generate \(M \) independent groups. For each group \(g\), the group size \(N_g\)  randomly determined between a preset minimum and maximum number of individuals. Within each group, we construct a small-world network using the Watts–Strogatz model. The number of nodes equals \(N_g\), the initial number of neighbors is \(k \in \{4, 8\}\) (depending on the dependence strength), and the rewiring probability is \(p \in \{0.1, 0.5\}\). The resulting adjacency matrix is denoted \(\mathbf{A}_g\).
	
	\begin{figure}[htbp]
		\centering
		\includegraphics[width=0.6\textwidth]{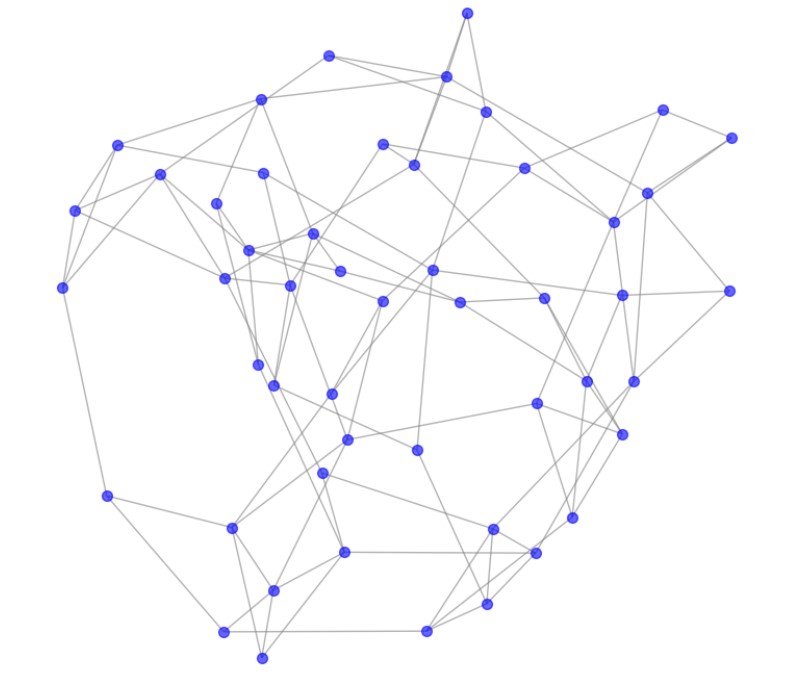} 
		\caption{Watts–Strogatz}
		\label{fig1}
		\footnotesize
		Note: Demonstrating the network topology with 50 nodes, 100 edges, and an average degree of 4.00.
	\end{figure}
	
	According to the exponential family condition (Assumption~\ref{ass4}), the joint distribution of \((\mathbf{W}_g,\mathbf{X}_g,\mathbf{A}_g)\) conditional on the latent group label \(L_g\) belongs to an exponential family with sufficient statistic \(S(\mathbf{W}_g,\mathbf{X}_g,\mathbf{A}_g)\).
	In our simulation, the unobserved group label \(L_g\) is operationalized via two group-specific latent variables, which are a group intercept \(\alpha_g\) that shifts the outcome equation and a covariate mean  \(\mu_{X,g}\) that shifts the distribution of individual covariates.
	These two variables together act as the empirical counterpart of \(L_g\). They capture all between-group heterogeneity that may confound the treatment effect.
	
	To vary the strength of group-level heterogeneity, we set the parameters as either low heterogeneity with \(\alpha_g \sim \mathcal{N}(0,1.5^2)\) and \(\mu_{X,g} \sim \mathcal{N}(0,1.0^2)\), or high heterogeneity with \(\alpha_g \sim \mathcal{N}(0,3.0^2)\) and \(\mu_{X,g} \sim \mathcal{N}(0,2.0^2)\).
	Individual covariates are then drawn as \(X_i \sim \mathcal{N}(\mu_{X,g(i)}, 1)\), independently across individuals. This specification ensures that the distribution of \(X_i\) varies across groups in a way. 

	To generate network confounding, we let the treatment probability depend on the individual’s own covariate, the average covariate of its neighbors, its degree, and the group-level mean shift \(\mu_{X,g}\) (which is part of the unobserved \(L_g\)). Specifically,
	\[
	\operatorname{logit}\bigl(P(W_i=1)\bigr) = \gamma_0 + \gamma_1 X_i + \gamma_2 \frac{\sum_j A_{ij} X_j}{\sum_j A_{ij}} + \gamma_3 \deg_i + \gamma_4 \mu_{X,g(i)},
	\]
	where \(\deg_i = \sum_j A_{ij}\). The intercept \(\gamma_0\) is calibrated so that the overall treatment proportion is approximately 0.5. The coefficients are set to \((\gamma_1,\gamma_2,\gamma_3,\gamma_4) = (0.5,0.5,0.2,0.5)\) for weak network confounding, and to \((1.5,1.5,0.8,1.5)\) for strong network confounding.

	The outcome \(Y_i\) is generated as
	\[
	Y_i = \alpha_{g(i)} + \beta X_i + \tau_{g(i)} \cdot T_i + \delta \cdot \frac{\sum_j A_{ij} W_j}{\sum_j A_{ij}} + \varepsilon_i,
	\]
	with \(\beta = 1.5\), \(\varepsilon_i \sim \mathcal{N}(0,0.5^2)\). The group-specific treatment effect \(\tau_{g}\) is drawn from \(\mathcal{N}(0.5, \sigma_\tau^2)\); we set \(\sigma_\tau = 0\) under low heterogeneity and \(\sigma_\tau = 0.15\) under high heterogeneity to allow for effect heterogeneity as permitted by our framework. The exposure mapping \(T_i\) is defined as the network exposure \(T_i = \mathbf{1}\{\sum_j A_{ij} W_j \ge 1\}\) (at least one treated neighbor), and the network interference coefficient \(\delta\) is set to \(0.8\) for weak network dependence and \(3.0\) for strong dependence.
	
	The graph neural network employs a Graph Convolutional Network architecture, utilizing a normalized adjacency matrix for message passing. The model consists of two \textit{GCN} convolutional layers, each with 16 or 32 hidden channels, and applies dropout (with a rate of 0.1 or 0.2) after each convolution to prevent overfitting. The hidden layers use \textit{ReLU} activation, while the output layer employs a linear activation for regression tasks or a Sigmoid for classification tasks. The model is trained using PyTorch's Adam optimizer with a learning rate of 0.001 or 0.005.
	
	The \textit{GNN}-based estimator \textit{(GME-GNN)} jointly estimates propensity scores and outcome regressions within each group. The \textit{GNN}-only estimator does not use any group-level sufficient statistics. The Mundlak estimator incorporates individual treatment, covariates, and their within-group means via \textit{OLS} regression to absorb between-group heterogeneity. These three methods serve as benchmarks for comparison with the \textit{GNN}-based estimator.
	
	\subsection{Simulation Under Weak Network Dependence}
	Table \ref{tab:l-w} shows that under weak network dependence and low between-group heterogeneity, \textit{GME-GNN} performs best across both group sizes. For $N_g = 100 \sim 200$, its \textit{RMSE} is 0.2883, outperforming \textit{GNN-only} (0.4998) and Mundlak (0.4075). When group size increases to $N_g = 300 \sim 400$, \textit{GME-GNN} 's \textit{RMSE} drops to 0.2068, and \textit{GNN-only} improves sharply to 0.2097; Mundlak's \textit{RMSE} remains high at 0.4327.
	
	
	
	\begin{table}[htbp]
		\centering
		\caption{Simulation results for Heterogeneity=low, Dependence=weak}
		\label{tab:l-w}
		\begin{tabular}{llcccccc}
			\toprule
			& & \multicolumn{3}{c}{$N_g=100\sim200$} & \multicolumn{3}{c}{$N_g=300\sim400$} \\
			$M$ & Method  & \textit{MAE} & \textit{MSE} & \textit{RMSE}  & \textit{MAE} & \textit{MSE} & \textit{RMSE}\\
			\midrule
			\multirow{3}{*}{100}
			& \textit{GME-GNN}   & 0.2502 & 0.0831 & 0.2883  & 0.2033 & 0.0482 & 0.2068  \\
			& \textit{GNN-only}  & 0.4152  & 0.2498 & 0.4998  & 0.2067 & 0.0440 & 0.2097 \\
			& Mundlak  & 0.4064 & 0.1660 & 0.4075  & 0.4326 & 0.1873 & 0.4327 \\
			\bottomrule
		\end{tabular}
		\vspace{0.3em}
		\parbox{\linewidth}{\footnotesize
			Notes: Results are based on 1,000 independent simulations. 
			\textit{MAE} (mean absolute error) $= \frac{1}{n} \sum_{i=1}^{n} |\hat{\tau}_i - \tau_i|$, 
			\textit{MSE} (mean squared error) $= \frac{1}{n} \sum_{i=1}^{n} (\hat{\tau}_i - \tau_i)^2$, 
			\textit{RMSE} (root mean squared error) $= \sqrt{\text{MSE}}$. 
			In the tables, $M = 100$ denotes the number of groups, and $N_g$ is the group size.}
	\end{table}

Table \ref{tab:h-w}  indicates that under high heterogeneity, \textit{GME-GNN} continues to be the best, with \textit{RMSE}s of 0.2670 and 0.1999 for the two group sizes. \textit{GNN-only} suffers a large error for smaller groups (0.6366), but its \textit{RMSE} plummets to 0.2056 when group size increases. Mundlak's error stays around \(0.42 \sim 0.43\) for both group sizes, showing little improvement.
\begin{table}[htbp]
	\centering
	\caption{Simulation results for Heterogeneity=high, Dependence=weak}
	\label{tab:h-w}
	\begin{tabular}{llcccccc}
		\toprule
		& & \multicolumn{3}{c}{$N_g=100\sim200$} & \multicolumn{3}{c}{$N_g=300\sim400$} \\
		$M$ & Method  & \textit{MAE} & \textit{MSE} & \textit{RMSE}  & \textit{MAE} & \textit{MSE} & \textit{RMSE}\\
		\midrule
		\multirow{3}{*}{100}
		& \textit{GME-GNN}   & 0.2320 & 0.0713 & 0.2670  & 0.1969 & 0.0400 & 0.1999  \\
		& \textit{GNN-only}  & 0.5057  & 0.4053 & 0.6366 & 0.2030  & 0.0423 & 0.2056 \\
		& Mundlak  & 0.4203 & 0.1783 & 0.4223  & 0.4304 & 0.1854 & 0.4306 \\
		\bottomrule
	\end{tabular}
	\vspace{0.3em}
	\parbox{\linewidth}{\footnotesize
		Notes: Results are based on 1,000 independent simulations. 
		\textit{MAE} (mean absolute error) $= \frac{1}{n} \sum_{i=1}^{n} |\hat{\tau}_i - \tau_i|$, 
		\textit{MSE} (mean squared error) $= \frac{1}{n} \sum_{i=1}^{n} (\hat{\tau}_i - \tau_i)^2$, 
		\textit{RMSE} (root mean squared error) $= \sqrt{\text{MSE}}$. 
		In the tables, $M = 100$ denotes the number of groups, and $N_g$ is the group size.}
\end{table}

Under conditions of weak network dependence, as inter-group heterogeneity increases from low to high, \textit{GME-GNN} consistently delivers robust and precise estimates, exhibiting the least susceptibility to the effects of heterogeneity. The \textit{GNN-only} model is significantly impacted by heterogeneity. While it performs reasonably well under low heterogeneity, it exhibits severe bias in small group settings when heterogeneity is high. The Mundlak approach is insensitive to heterogeneity, but its estimation error consistently exceeds that of \textit{GME-GNN}, and increasing the group size does not significantly improve its performance.

\subsection{Simulation Under Strong Network Dependence}

Table \ref{tab:l-s} reports the simulation results under strong network dependence and low between-group heterogeneity. The \textit{GME-GNN} substantially outperforms both \textit{GNN-only} and the Mundlak estimator. For $N_g = 100 \sim 200$, \textit{GME-GNN} achieves an \textit{RMSE} of 0.5706, which is much lower than that of \textit{GNN-only} (0.8751) and Mundlak (0.9563). When the group size increases to $N_g = 300 \sim 400$, the \textit{RMSE} of \textit{GME-GNN} drops to 0.4366, whereas those of \textit{GNN-only} and Mundlak remain high at 0.9069 and 0.9580, respectively, showing little improvement.
\begin{table}[htbp]
	\centering
	\caption{Simulation results for Heterogeneity=low, Dependence=strong}
	\label{tab:l-s}
	\begin{tabular}{llcccccc}
		\toprule
		& & \multicolumn{3}{c}{$N_g=100\sim200$} & \multicolumn{3}{c}{$N_g=300\sim400$} \\
		$M$ & Method  & \textit{MAE} & \textit{MSE} & \textit{RMSE}  & \textit{MAE} & \textit{MSE} & \textit{RMSE}\\
		\midrule
		\multirow{3}{*}{100}
		& \textit{GME-GNN}   & 0.5298 & 0.3255 & 0.5706  & 0.4083 & 0.1906 & 0.4366  \\
		& \textit{GNN-only}  & 0.7643 & 0.7658 & 0.8751  & 0.8307 & 0.8255 & 0.9069 \\
		& Mundlak  & 0.9624 & 0.9146 & 0.9563  & 0.9538 & 0.9177 & 0.9580 \\
		\bottomrule
	\end{tabular}
	\vspace{0.3em}
	\parbox{\linewidth}{\footnotesize
		Notes: Results are based on 1,000 independent simulations. 
		\textit{MAE} (mean absolute error) $= \frac{1}{n} \sum_{i=1}^{n} |\hat{\tau}_i - \tau_i|$, 
		\textit{MSE} (mean squared error) $= \frac{1}{n} \sum_{i=1}^{n} (\hat{\tau}_i - \tau_i)^2$, 
		\textit{RMSE} (root mean squared error) $= \sqrt{\text{MSE}}$. 
		In the tables, $M = 100$ denotes the number of groups, and $N_g$ is the group size.}
\end{table}

Table \ref{tab:h-s} shows the results under high heterogeneity and strong network dependence, where all methods exhibit larger errors. For $N_g = 100 \sim 200$, \textit{GME-GNN} has an \textit{RMSE} of 0.9309, compared to 1.8120 for \textit{GNN-only} and 1.1131 for Mundlak. When the group size increases to $N_g = 300 \sim 400$, the \textit{RMSE} of \textit{GME-GNN} declines sharply to 0.5788, making it the best-performing method. In contrast, the \textit{RMSE} of \textit{GNN-only} rises to 2.4869, indicating severe performance deterioration, while that of Mundlak decreases slightly to 1.0213, still significantly higher than \textit{GME-GNN}.

\begin{table}[htbp]
	\centering
	\caption{Simulation results for Heterogeneity=high, Dependence=strong}
	\label{tab:h-s}
	\begin{tabular}{llcccccc}
		\toprule
		& & \multicolumn{3}{c}{$N_g=100\sim200$} & \multicolumn{3}{c}{$N_g=300\sim400$} \\
		$M$ & Method  & \textit{MAE} & \textit{MSE} & \textit{RMSE}  & \textit{MAE} & \textit{MSE} & \textit{RMSE}\\
		\midrule
		\multirow{3}{*}{100}
		& \textit{GME-GNN}   & 0.9309 & 0.8665 & 0.9309  & 0.4906 & 0.3350 & 0.5788  \\
		& \textit{GNN-only}  & 1.8120 & 3.2834 & 1.8120  & 1.8452 & 6.1845 & 2.4869 \\
		& Mundlak  & 1.1131 & 1.2389 & 1.1131  & 1.0109 & 1.0431 & 1.0213 \\
		\bottomrule
	\end{tabular}
	\vspace{0.3em}
	\parbox{\linewidth}{\footnotesize
		Notes: Results are based on 1,000 independent simulations. 
		\textit{MAE} (mean absolute error) $= \frac{1}{n} \sum_{i=1}^{n} |\hat{\tau}_i - \tau_i|$, 
		\textit{MSE} (mean squared error) $= \frac{1}{n} \sum_{i=1}^{n} (\hat{\tau}_i - \tau_i)^2$, 
		\textit{RMSE} (root mean squared error) $= \sqrt{\text{MSE}}$. 
		In the tables, $M = 100$ denotes the number of groups, and $N_g$ is the group size.}
\end{table}

Under conditions of strong network dependence, estimation methods face stricter demands. At low heterogeneity, \textit{GME-GNN} significantly outperforms the other two methods, and increasing group size helps improve its accuracy, while the errors of \textit{GNN‑only} and Mundlak remain persistently high. At high heterogeneity, \textit{GME‑GNN} remains the most robust performer. The \textit{GNN‑only} method exhibits extreme instability and is heavily influenced by heterogeneity, whereas the Mundlak approach is relatively less affected by heterogeneity, yet its overall error is significantly higher than that of \textit{GME‑GNN}.


\section{Empirical Application}
\subsection{Research Context and Data Source}
This study empirically examines the causal effects of China’s high-tech enterprise certification policy. This policy serves as a key instrument for promoting innovation and industrial upgrading. It works by providing multiple incentives, including reducing the corporate income tax rate for certified enterprises from 25\% to 15\%, allowing additional tax deductions for R\&D expenses, and offering financial support such as credit assistance and special subsidies. Through the dual mechanisms of cost compensation and profit enhancement, these policies directly boost enterprises' R\&D intensity, innovation efficiency, and total factor productivity. In contrast to traditional, universally-applied policies, this approach employs certification as a filtering mechanism. Consequently, its influence is not confined to the certified enterprises. Instead, it radiates outward through technology diffusion, factor flows, and supply chain networks, reaching both upstream and downstream players in the industrial chain.
However, the certification status of firms is non-random; it is systematically linked to factors such as firm size, and governance structure, technological capability.
{\tiny }
This study employs cross-sectional data from 2024 for empirical analysis, with the sample consisting of A-share listed companies in China. Core treatment variables and categorical information, such as the firm's high-tech enterprise certification status and industry classification, are sourced from the Wind Financial Database. Financial indicators and governance characteristics, including return on assets (\textit{ROA}), firm size, asset-liability ratio, cash flow ratio, and firm age, are obtained from the \textit{CSMAR} Database. To map network dependence structures, data on upstream and downstream supply chain relationships between firms are sourced from the \textit{RESSET} Financial Research Database. After removing observations with missing key variables, abnormal listing status, and outliers, the final cross-sectional sample comprises 4693 listed companies. \textit{ROA} is the main measure of financial performance. Key firm-level control variables such as firm size, leverage ratio, and cash flow ratio are also included.

\subsection{Group Definition and Network Construction}
This study classifies the sample according to the primary industry categories specified in the Guidelines for the Industry Classification of Listed Companies(2012) issued by the China Securities Regulatory Commission, initially dividing it into 18 industry groups. We excluded industries with too few firms to guarantee sufficient sample sizes, robust statistical estimates, and the reliable identification of network structures.  Specifically, industries with codes $A$ (Agriculture, Forestry, Animal Husbandry and Fishery), $H$ (Accommodation and Catering Services), $O$ (Resident Services, Repair, and Other Services), $P$ (Education), $Q$ (Health and Social Work), and $S$ (Comprehensive Industries) each contained fewer than 50 firms in the sample. Such small group sizes could compromise the precision of within-group estimates and hinder the identification of inter-group heterogeneity and within-group network dependence structures. After removing these industries, the retained groups provide an adequate sample base, which not only helps control for unobserved heterogeneity across groups but also offers a reliable data foundation for characterizing within-group network dependencies.

To capture within-group interactions and dependencies, a supply chain network is built for each retained industry group using actual inter-firm economic ties.  
For each retained industry group, we construct an undirected binary adjacency matrix \( \mathbf{A}_g \), where firms within the industry serve as nodes and the presence (or absence) of an actual supply chain relationship between firms defines the edges. The element $A_{ij}$ of matrix $A$ is assigned a value of 1 if there exists an upstream or downstream relationship between firm $i$ and firm $j$; otherwise, it takes a value of 0.
By constructing supply chain networks within each industry group, we are able to control for unobserved heterogeneity at the industry level while precisely capturing within-group spillover effects and network interference.

\subsection{Empirical Results}
We need to clarify the set of balancing scores \( S(\mathbf{X}_g, \mathbf{W}_g, \mathbf{A}_g) \) used to support the unconfoundedness assumption. Therefore, \( \bar{S}_{g(i)} \) is implicitly defined as \( \bar{S}_{g(i)} = ( \bar{W}_g, \bar{X}_g, \overline{A_g W_g}, \\
\overline{A_g X_g}) \). To accurately measure the direct and spillover effects of high-tech certification on firm performance, we consider the estimator \( \tau(t, t') \) defined by the exposure mapping \( T_i = \left( W_i, \mathbf{1} \left\{ \sum_{j=1}^n A_{ij} W_j > 0 \right\} \right) \). According to this exposure mapping, firms are divided into four exposure levels: $T=0$ (untreated with no treated neighbors), $T=1$ (untreated with treated neighbors), $T=2$ (treated with no treated neighbors), and $T=3$ (treated with treated neighbors). Specifically, the comparison between $T=0$ and $T=1$ is used to identify the pure spillover effect, while the comparison between $T=0$ and $T=2$ is adopted to estimate the direct effect. 

The distribution of exposure levels across industries is shown in Figure \ref{fig2}. To ensure estimation validity, industries with insufficient observations are excluded, including Industry B and Industry K. As can be seen from Figure \ref{fig2}, the number of firms belonging to the $T=3$ group is relatively small, which cannot meet the data requirements for effective group comparison and parameter identification. Therefore, this study does not further conduct effect identification and empirical analysis for the $T=3$ group.
\begin{figure}[htbp]
	\centering
	\includegraphics[width=0.6\textwidth]{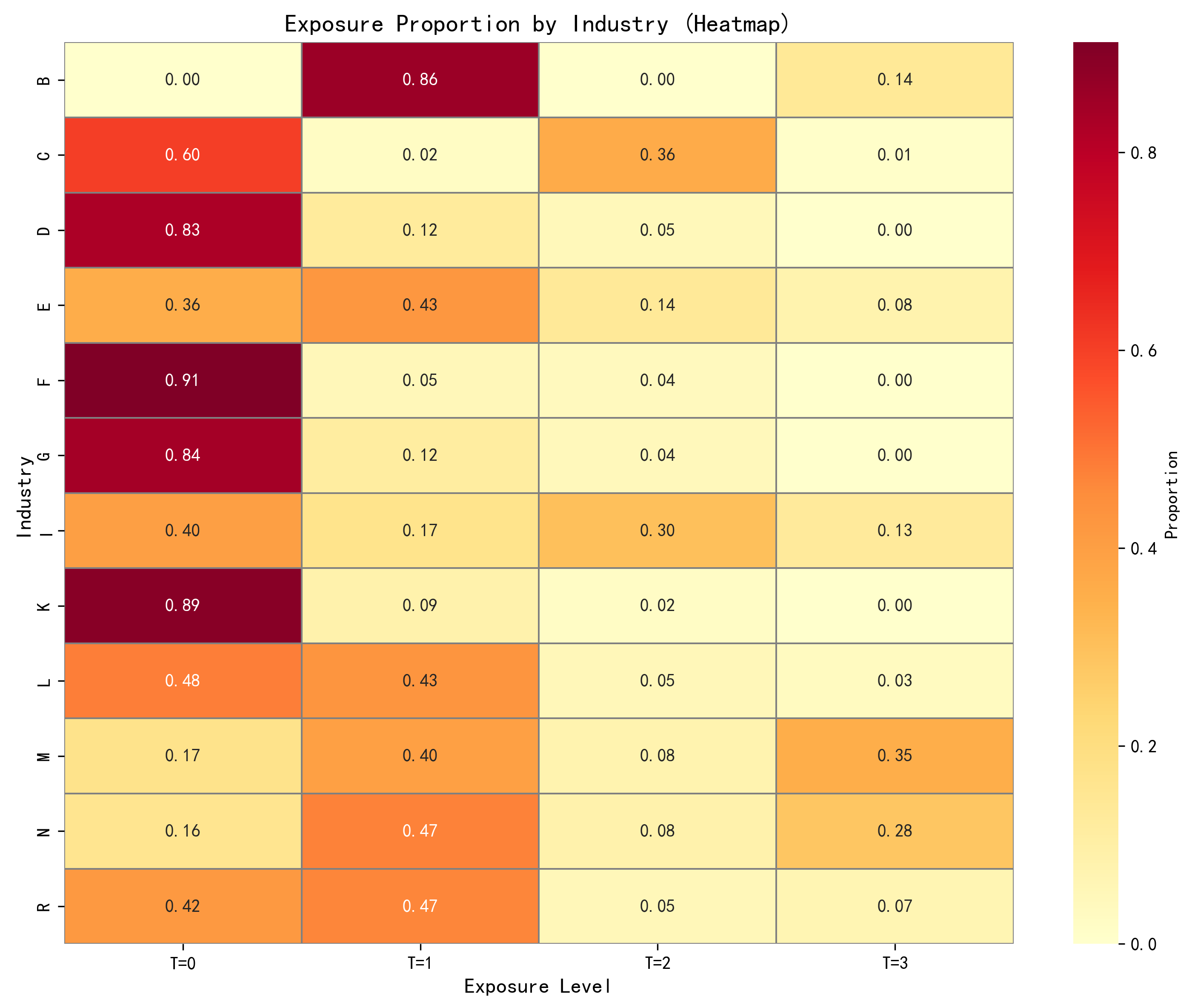}
	\caption{Industry Exposure Distribution
	}
	\label{fig2}
\end{figure}

Based on the balancing scores $\bar{S}_{g(i)}$, this study employs a \textit{GNN} to estimate propensity scores and conditional mean outcomes, with inverse probability weights truncated to $[0.01, 0.99]$ to mitigate extreme-value bias. Figure \ref{fig3} presents the kernel density distributions of estimated propensity scores across the four exposure levels. The distributions of $T=0$ and $T=1$ overlap substantially within $[0.25, 0.45]$, satisfying the common support condition for estimating the pure spillover effect. The distributions of $T=0$ and $T=2$ exhibit considerable overlap within $[0.20, 0.40]$, enabling the identification of the direct effect. The distribution of $T=3$ peaks sharply at $0.20 \sim 0.25$, nearly coinciding with that of $T=2$, and shows limited common support with other groups; therefore, spillover effects on treated units are not separately identified.
\begin{figure}[htbp]
	\centering
	\includegraphics[width=0.6\textwidth]{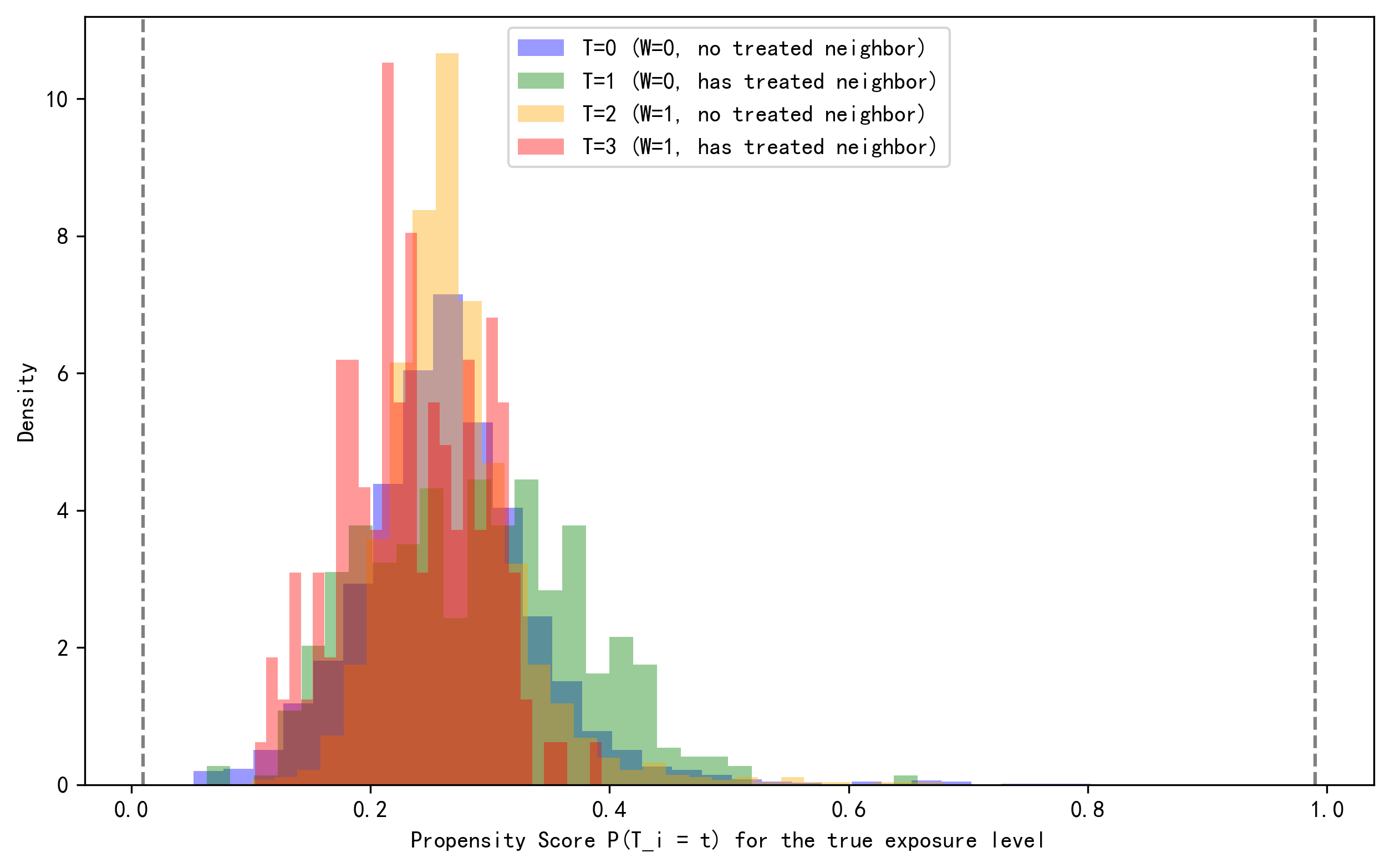} 
	\caption{Propensity Score Distribution by Treatment Status}
	\label{fig3}
\end{figure}
Table \ref{tab:effects} reports the estimated effects. The \textit{GME}-\textit{GNN} estimator yields a direct effect of 0.0274 and a spillover effect of 0.0182, both significant at the 0.1\% level. High-tech certification improves certified firms' ROA by 2.74 percentage points, while untreated firms with treated neighbors gain 1.82 percentage points through supply-chain spillovers. The \textit{GNN}-only estimator produces a larger direct effect (0.0908), confirming that the Mundlak-type balancing statistics effectively control for between-group heterogeneity.  We do not report the traditional Mundlak estimator here because it relies on the \textit{SUTVA} and cannot accommodate within-group network interference.

\begin{table}[h!]
	\centering
	\caption{Estimated Direct and Spillover Effects}
	\label{tab:effects}
	\small
	\begin{tabular}{lcc}
		\toprule
		& Direct Effect & Spillover Effect \\
		\midrule
		GME-GNN & 0.0274$^{***}$ & 0.0182$^{***}$ \\
		& (0.0008) & (0.0008) \\
		\addlinespace
		GNN-only & 0.0908$^{***}$ & 0.0251$^{***}$ \\
		& (0.0032) & (0.0004) \\
		\bottomrule
		\multicolumn{3}{l}{\footnotesize Notes: Standard errors in parentheses. $^{***}$p$<$0.001.}
	\end{tabular}
\end{table}

\section{Conclusion}
\label{sec:conc}
This paper proposes a generalized Mundlak estimator based on graph neural networks (\textit{GME}\text{-}\textit{GNN}) to address causal inference under network confounding, where both between-group heterogeneity and within-group dependence are present. The proposed estimator combines the strength of the Mundlak approach in controlling between-group confounding with the ability of graph neural networks to capture nonlinear dependencies and network interference. Theoretical analysis shows that, under appropriate assumptions, the \textit{GME}\text{-}\textit{GNN} estimator satisfies double robustness and asymptotic normality. A network \textit{HAC} variance estimator is employed to facilitate valid inference.

Simulation studies have verified that under varying levels of inter-group heterogeneity and network dependence, \textit{GME}\text{-}\textit{GNN} achieves lower root mean square error and stronger robustness compared to the traditional Mundlak estimator and pure \textit{GNN} methods.
An empirical analysis of the impact of China's high-tech enterprise certification policy on the financial performance of A‑share listed companies further demonstrates the effectiveness of the method. The \textit{GME}\text{-}\textit{GNN} estimates a direct effect of 2.74 percentage points and a spillover effect of 1.82 percentage points. This result highlights the model’s advantage in controlling for industry‑level confounders and capturing spillover effects through supply‑chain networks. Overall, \textit{GME}\text{-}\textit{GNN} provides a flexible, robust, and interpretable tool for causal inference in observational studies with network interference and group-level heterogeneity.


 
\section{Disclosure statement}\label{disclosure-statement}
No competing interest is declared.

\section{Data Availability Statement}
The data analyzed in this study are obtained from public databases. Data are available from the corresponding author upon reasonable request.

\bibliography{bibliography.bib}

\end{document}